\documentclass[11pt]{article}

\usepackage[a4paper,margin=1in]{geometry}
\usepackage[T1]{fontenc}
\usepackage[utf8]{inputenc}
\usepackage{mathpazo}
\usepackage{microtype}

\usepackage{amsmath,amssymb,amsthm,mathtools}
\usepackage{mathpartir}
\usepackage{enumitem}
\usepackage{float}
\usepackage{needspace}
\usepackage{aliascnt}

\usepackage[hidelinks]{hyperref}
\usepackage[nameinlink,noabbrev]{cleveref}

\newcommand{\IKfour}{\mathsf{IK4}}
\newcommand{\NIKfour}{\mathsf{NIK4}}
\newcommand{\Sub}{\operatorname{Sub}}
\newcommand{\Up}{\mathord{\uparrow}}
\newcommand{\Pre}{\operatorname{Pre}}

\theoremstyle{plain}
\newtheorem{theorem}{Theorem}[section]
\newaliascnt{lemma}{theorem}
\newtheorem{lemma}[lemma]{Lemma}
\aliascntresetthe{lemma}
\newaliascnt{proposition}{theorem}
\newtheorem{proposition}[proposition]{Proposition}
\aliascntresetthe{proposition}
\newaliascnt{corollary}{theorem}
\newtheorem{corollary}[corollary]{Corollary}
\aliascntresetthe{corollary}

\theoremstyle{definition}
\newaliascnt{definition}{theorem}
\newtheorem{definition}[definition]{Definition}
\aliascntresetthe{definition}

\title{A Kruskal Decision Procedure for Intuitionistic Modal Logic $\IKfour$}
\author{Mario Piazza\\
\small Scuola Normale Superiore, Pisa\\
\small \href{mailto:mario.piazza@sns.it}{mario.piazza@sns.it}}
\date{}

\begin{document}
\maketitle

\begin{abstract}
We prove decidability of Simpson's intuitionistic modal logic $\IKfour$ by
working directly with cut-free nested proofs. Once the end formula is fixed,
only finitely many combinations of input and output formulae can occur at a
node, although the modal tree itself remains unbounded. We order these nested sequents by rooted homeomorphic embedding: weakening may
add input formulae, while transitivity allows a modal edge to be stretched
into a non-empty path.

Kruskal's theorem makes rooted homeomorphic embedding a well-quasi-order, but
does not by itself make backward application of the rules effective: an
inference may still occur inside an arbitrarily large context. The finite-support lemma shows that a minimal predecessor need retain only the
positions used by the inference, the images of the chosen basis elements, and
the branch points joining them.
Together with an effective enumeration of bounded rule instances, this
bound makes the minimal predecessors computable. Backward closure from the initial sequents gives an increasing sequence
of finitely based upward-closed sets. The sequence eventually stabilises, and
its stable value is the set of provable nested sequents.  At that point,
finitely many cut-free proofs suffice: every other provable nested sequent is
obtained from one of them by weakening along an embedding.  Their maximum
height gives a uniform proof-height bound.
\end{abstract}

\medskip
\noindent\textbf{2020 Mathematics Subject Classification.}
Primary 03B45; Secondary 03B25, 03F03.

\smallskip
\noindent\textbf{Keywords.}
Intuitionistic modal logic, decidability, nested sequents, well-quasi-orders,
Kruskal's theorem.

\section{Introduction}

Fischer Servi's intuitionistic modal logics combine an intuitionistic
preorder with an independent modal relation, coordinated by two commuting
conditions~\cite{FischerServi1984}.
Simpson later developed the proof-theoretic and birelational presentation
used here~\cite{Simpson1994}.  For each formula $A$, the transitive extension
$\IKfour$ adds
\[
 \Box A\to\Box\Box A,
 \qquad
 \Diamond\Diamond A\to\Diamond A.
\]
The semantic clauses are familiar; the problem addressed here is the
termination of proof search.

Earlier work already separates the availability of analytic calculi from this
termination problem.  Amati and
Pirri already included $\IKfour$ in their uniform tableau and sequent
treatment~\cite{AmatiPirri1994}, and Simpson obtained strong normalization and
cut-free sequent calculi throughout his family.  Yet decidability did not
follow uniformly.  Simpson's procedure covered the systems listed in his Fig.~7--5, but left
$\IKfour$, $\mathsf{IKD4}$, and $\mathsf{IS4}$ open
\cite[Fig.~7--5 and pp.~146--147]{Simpson1994}.
Transitivity was already causing the relevant difficulty: modal requirements
can be propagated arbitrarily far down a graph extension, and the extensions
used to control backward search are no longer bounded.  Nested and deep
sequent methods for modal and tense logics already have a substantial
proof-theoretic history~\cite{Brunnler2009,GorePostnieceTiu2011,TiuIanovskiGore2012}.
In the intuitionistic setting, Kuznets and Stra{\ss}burger covered the whole
modal cube and isolated the non-termination peculiar to the transitive and
Euclidean cases~\cite{KuznetsStrassburger2019}.  Lyon then used structural
refinement to obtain the cut-free propagation
calculus used below, while still describing decidability of transitive
extensions of $\mathsf{IK}$ as a longstanding open
problem~\cite[p.~424]{Lyon2021}.

Several earlier decidability results come close enough to require a precise
boundary.  Alechina and Shkatov obtain decidability when the relevant frame
conditions can be presented as an acyclic family of monadic-second-order
definable closure conditions, subject to their condition on the number of
closure conditions associated with each relation.  This does not cover
Simpson's frames: in their notation they explicitly give
$R_2\circ R\subseteq R\circ R_2$ as a non-example for their method
\cite[p.~229]{AlechinaShkatov2006}.  Identifying $R$ with the intuitionistic
preorder and $R_2$ with the modal accessibility relation, this is exactly
the commuting condition F1 used here.

Garg, Genovese and Negri obtained a constructive decision procedure for the
necessitation-only fragment of intuitionistic $\mathsf{K4}$, and explicitly
left possibility modalities outside the scope of their intuitionistic method
\cite{GargGenoveseNegri2012}.  Their result therefore does not decide
Simpson's full $\IKfour$, whose language contains both $\Box$ and
$\Diamond$, and whose transitive extension includes both
$\Box A\to\Box\Box A$ and $\Diamond\Diamond A\to\Diamond A$.
Voorneveld's recent forward procedure is likewise formulated for unary
$\mathsf K$-modalities and transformation axioms between them, including
axiom $4$; it does not add the intuitionistic possibility operator required
for the full Simpson system~\cite{Voorneveld2025}.

Galmiche and Salhi give decision procedures for $\mathsf{IK}$,
$\mathsf{IT}$, $\mathsf{IB4}$, and $\mathsf{ITB}$; $\IKfour$ is not
among the systems for which they establish decidability
\cite{GalmicheSalhi2018}.  A recent contribution proves decidability of
$\mathbf{LIK4}$~\cite{BalbianiGencerTinchev2025}, but this is a different
logic: $\mathbf{LIK}$ is incomparable with Simpson's $\mathsf{IK}$, and
its box and diamond are interpreted locally on frames satisfying both forward
and downward confluence~\cite{BalbianiGaoGencerOlivetti2024}.  Thus analytic
calculi for Simpson's basic transitive system were already available; what
remained open was a terminating decision procedure for the full
$\Box$--$\Diamond$ logic.

Filtration does not immediately settle the matter either.  Quotienting a
transitive birelational model by a simulation need not preserve transitivity,
while taking the transitive closure of the induced modal relation may introduce
modal obligations absent from the original model.  Simpson's finite-model
construction stops short of the transitive case because the relevant contexts
can no longer be bounded~\cite[pp.~169, 174--175]{Simpson1994}.  The
example after the semantic clauses shows where this direct construction
breaks.

Here the choice of calculus matters.  We use Lyon's monomodal $\IKfour$
calculus because structural refinement
exposes the part of backward search that remains unbounded.  Transitivity is
represented by propagation along the modal tree already present in the
sequent, rather than by a structural rule that changes that tree.  In the
set-based presentation used below, fixing the end formula leaves only the
modal tree unbounded.  We import the labelled completeness theorem and the structural-refinement
results stated below, spell out the passage to the nested presentation used
here, and prove termination there.

The calculus itself suggests the order to use on nested sequents.  Weakening permits
additional input formulae, while transitivity makes the subdivision of a modal
edge harmless.  This leads to the rooted homeomorphic order on nested sequents.  Derivability is monotone for this order.  Kruskal's labelled Tree
Theorem~\cite{Kruskal1960} then rules out infinite bad sequences;
Nash--Williams's formulation makes explicit the underlying relation of
homeomorphism onto a subtree~\cite{NashWilliams1963}.

Kruskal gives finite bases for the resulting upward-closed sets.  It does
not, however, tell us how to compute backward rule applications inside
unbounded contexts.  As in backward coverability for well-structured transition
systems~\cite{FinkelSchnoebelen2001}, the well-quasi-order is algorithmically
useful only once predecessor bases can be computed.  The finite-support lemma
below supplies the bound needed to compute them for these rules.

Consider a single inference.  A minimal conclusion need retain only the
positions used by the rule, the images of the chosen basis elements, and the
branch points joining them.  Side branches outside this
finite part may be deleted, and unmarked vertices of degree $2$ suppressed.
The resulting bound reduces predecessor computation to a finite search;
every candidate can be checked effectively against the displayed rules.
Backward closure from the initial sequents can then be carried out on finite
bases, producing an increasing chain whose eventual value is the set of
provable nested sequents.

\section{Semantic and proof-theoretic setting}

Let $p$ range over propositional variables.  The language is generated by
\[
 A ::= p\mid\bot\mid A\wedge A\mid A\vee A\mid A\to A
       \mid \Box A\mid\Diamond A.
\]
A birelational frame is $(W,\leq,R)$, where $W$ is non-empty, $\leq$ is a preorder and
\begin{align*}
\tag{F1}
 xRy\ \&\ y\leq z
 &\Longrightarrow \exists u\,(x\leq u\ \&\ uRz),\\
\tag{F2}
 x\leq y\ \&\ xRz
 &\Longrightarrow \exists u\,(yRu\ \&\ z\leq u).
\end{align*}
A frame for $\IKfour$ additionally has transitive $R$.  Valuations are upward
closed along $\leq$.  Write $x\Vdash A$ when $A$ is true at $x$; the modal
clauses are
\begin{align*}
 x\Vdash\Diamond A
 &\Longleftrightarrow \exists y\,(xRy\ \&\ y\Vdash A),\\
 x\Vdash\Box A
 &\Longleftrightarrow
 \forall y,z\,(x\leq y\ \&\ yRz\Rightarrow z\Vdash A).
\end{align*}

These are Fischer Servi's birelational clauses in Simpson's formulation~\cite{FischerServi1984,Simpson1994}.  Persistence of $\Diamond$ uses F2.  Persistence of $\Box$ is already built into its universal clause; F1 will instead be used when an assignment is lifted against the direction of a modal edge.

\medskip
\noindent\emph{Example.}
Consider the finite frame on $a,u,v,y,t$ in which $\leq$ is reflexive with the
single additional comparison $a\leq u$, and
\[
  R=\{(u,v),(y,t)\}.
\]
The relation $R$ is transitive.  Condition F1 holds because the target of
either $R$-edge has no proper $\leq$-extension, so the same source witnesses
the conclusion.  For F2, the only non-reflexive comparison is $a\leq u$, and
there is no $R$-edge with source $a$; the remaining cases are witnessed by the
original target.  Let $p$ be true at $v$ and $y$, and nowhere else.  The
valuation is upward closed.  By the box clause, $\Box p$ is true at $a$,
while $p$ is false at $t$.

Suppose that a finite quotient identifies $v$ and $y$ because the invariant
used by the quotient does not retain their different modal incidences.  Before
quotienting, $yRt$ is irrelevant to the box at $a$: there is no chain
$a\leq y$ and $yRt$.  In the quotient, however,
\[
  [u]R[v],
  \qquad [v]=[y],
  \qquad [y]R[t].
\]
The induced relation is no longer transitive; its transitive closure adds
$[u]R[t]$.  Since $[a]\leq[u]$, the closure has produced
$[a]\leq[u]R[t]$, and the box at $[a]$ now demands $p$ at $[t]$.  The
quotient has forgotten which representative carried which modal incidence.

This does not rule out a more careful finite-model construction.  It shows
only why the procedure of first quotienting the frame and then taking the
transitive closure of the induced modal relation is unsafe: identification can
erase the provenance of $R$-edges, while closure can turn the resulting
composable edges into new box obligations.
\medskip

\begin{lemma}\label{lem:persistence}
If $x\leq y$ and $x\Vdash A$, then $y\Vdash A$.
\end{lemma}

\begin{proof}
The proof is by induction on $A$.  The propositional cases are standard.  For
implication, if $x\Vdash B\to C$, $x\leq y\leq z$, and $z\Vdash B$, then
$z\Vdash C$ by the clause at $x$.

If $x\Vdash\Diamond B$, choose $z$ with $xRz$ and $z\Vdash B$.  From
$x\leq y$ and $xRz$, F2 gives $u$ with $yRu$ and $z\leq u$.  The induction
hypothesis yields $u\Vdash B$, hence $y\Vdash\Diamond B$.  Finally, suppose
$x\Vdash\Box B$ and $x\leq y$.  If $y\leq y'$ and $y'Rz$, then
$x\leq y'$, so the box clause at $x$ gives $z\Vdash B$.
\end{proof}

\begin{lemma}[Transitivity and axiom $4$]
\label{lem:semantic-four}
On every Fischer--Servi frame with transitive $R$, both
\[
 \Box A\to\Box\Box A
 \qquad\text{and}\qquad
 \Diamond\Diamond A\to\Diamond A
\]
are valid.
\end{lemma}

\begin{proof}
It is enough to establish the two pointwise implications at an arbitrary
world.  Indeed, if $x\Vdash B$ implies $x\Vdash C$ for every $x$, then every
world forces $B\to C$: at an arbitrary $\leq$-extension one applies the same
pointwise implication.

Fix $x$ and assume $x\Vdash\Box A$.  To prove $x\Vdash\Box\Box A$, let
$x\leq y$ and $yRz$; it remains to show $z\Vdash\Box A$.  Take
$z\leq z'$ and $z'Ru$.  By F1, applied to $yRz\leq z'$, there is $y'$ with
$y\leq y'$ and $y'Rz'$.  Transitivity of $R$ gives $y'Ru$.  Since
$x\leq y\leq y'$ and $x\Vdash\Box A$, the box clause yields $u\Vdash A$.
Thus $z\Vdash\Box A$, and hence $x\Vdash\Box\Box A$.

For the diamond principle, assume $x\Vdash\Diamond\Diamond A$ and choose $y,z$ with
$xRy$, $yRz$, and $z\Vdash A$.  Transitivity gives $xRz$, so
$x\Vdash\Diamond A$.
\end{proof}

We write $\IKfour$ for the set of formulas valid on these frames.  Simpson's
completeness theorem shows that these are exactly the theorems of the
Hilbert-style extension of intuitionistic modal $\mathsf{K}$ by the two
displayed principles of axiom $4$~\cite{Simpson1994}.

\subsection{Why the nested presentation matters}

What matters here is not merely that Lyon's calculus is cut free.  After
structural refinement, transitivity becomes propagation along a path already
present in the nested sequent; it no longer appears as a structural rule that
alters the modal tree.  The unbounded part of backward search is therefore
visible in the proof object itself.  Once the end formula is fixed and input
repetitions are erased, node labels range over a finite set, while the tree may
still grow.  This is the object ordered by homeomorphic embedding below.

The predecessor argument uses precisely this locality.  Each
inference acts at finitely many nodes, and propagation refers to the rest
of the tree only through strict descendancy.  We display the rules in full
because the finite-support argument uses this local shape, especially what
happens to the conclusion output in the first premise of $\to^{\bullet}$.

\subsection{The nested calculus NIK4}

The schemes below are the monomodal $\IKfour$ instance of Lyon's refined
nested system~\cite[pp.~421--424]{Lyon2021}.  Lyon's nested presentation reads input contexts as finite multisets.
After importing completeness at the labelled level, we pass through this
multiset presentation before repetitions are erased.  For the decision argument we erase repeated input
occurrences node by node; \cref{lem:repetitions} shows that this transforms
every multiset proof into a proof in the set-based calculus without increasing
its height.  Formula endsequents are unchanged by the erasure.  Once the
subformula vocabulary is fixed, set-valued inputs leave only finitely many
possible node labels for the Kruskal encoding.

We use signed formulae $A^{\bullet}$ and $A^{\circ}$, called input and output
formulae.  Output-free nested sequents $\Delta$ and full nested sequents
$\Sigma$ are generated simultaneously by the following grammar, where
$n,k\geq0$:
\[
 \Delta ::= A_1^{\bullet},\ldots,A_n^{\bullet},
             [\Delta_1],\ldots,[\Delta_k],
 \qquad
 \Sigma ::= \Delta,A^{\circ}\mid \Delta,[\Sigma].
\]
Thus a full nested sequent is a finite rooted modal tree with a finite set of
input formulae at each node and exactly one output occurrence in the whole
tree.  In the second clause for $\Sigma$, the displayed full child is the
unique child containing that output; every other child belongs to $\Delta$.
The comma is associative and commutative.  Input formulae at one node are
treated as a set; sibling subtrees form a finite unordered family and their
multiplicity is retained.

An input or output context is a nested sequent with holes of the corresponding
polarity.  Filling an input hole is interpreted by set union: the fixed part of
the context may already contain the formula inserted at the hole.  A displayed
conclusion therefore does not by itself determine the inference; the fixed part
of the context must also be chosen.  In particular, when a principal input
formula is already present at the active node, two applications with the same
displayed conclusion may differ over whether that formula remains in the
premise.  In Lyon's system the same context is read by multiset union.  We
call that the multiset presentation.  When the distinction matters, we keep
one multiset realization in the background; otherwise we use the displayed set
notation.  If $\mathcal C$ contains the unique
output, $\mathcal C^{\downarrow}$ denotes the same context with that output
deleted, and is therefore output-free.

For nodes $u,v$ of a nested sequent, write $u\prec v$ when $v$ is a strict
modal descendant of $u$.  The calculus $\NIKfour$ consists of the following
deep rules.  Rules may be applied inside any context of the displayed polarity.

\begin{figure}[ht]
\centering
\begin{mathpar}
\inferrule*[right=$id$]{ }{\mathcal C\{p^{\bullet},p^{\circ}\}}
\and
\inferrule*[right=$\bot^{\bullet}$]{ }{\mathcal C\{\bot^{\bullet}\}}
\\
\inferrule*[right=$\wedge^{\bullet}$]
 {\mathcal C\{A^{\bullet},B^{\bullet}\}}
 {\mathcal C\{(A\wedge B)^{\bullet}\}}
\and
\inferrule*[right=$\wedge^{\circ}$]
 {\mathcal C\{A^{\circ}\}\\\mathcal C\{B^{\circ}\}}
 {\mathcal C\{(A\wedge B)^{\circ}\}}
\\
\inferrule*[right=$\vee^{\bullet}$]
 {\mathcal C\{A^{\bullet}\}\\\mathcal C\{B^{\bullet}\}}
 {\mathcal C\{(A\vee B)^{\bullet}\}}
\and
\inferrule*[right=$\vee_i^{\circ}$]
 {\mathcal C\{A_i^{\circ}\}}
 {\mathcal C\{(A_1\vee A_2)^{\circ}\}}
 \quad i\in\{1,2\}
\\
\inferrule*[right=$\to^{\circ}$]
 {\mathcal C\{A^{\bullet},B^{\circ}\}}
 {\mathcal C\{(A\to B)^{\circ}\}}
\and
\inferrule*[right=$\to^{\bullet}$]
 {\mathcal C^{\downarrow}\{(A\to B)^{\bullet},A^{\circ}\}
  \\
  \mathcal C\{B^{\bullet}\}}
 {\mathcal C\{(A\to B)^{\bullet}\}}
\\
\inferrule*[right=$\Diamond^{\bullet}$]
 {\mathcal C\{[A^{\bullet}]\}}
 {\mathcal C\{(\Diamond A)^{\bullet}\}}
\and
\inferrule*[right=$\Box^{\circ}$]
 {\mathcal C\{[A^{\circ}]\}}
 {\mathcal C\{(\Box A)^{\circ}\}}
\end{mathpar}
\caption{The logical rules of $\NIKfour$.}
\label{fig:nik4-logical}
\end{figure}

The two remaining modal rules are propagation rules.  The notation
$\mathcal C\{\cdot\}_u\{\cdot\}_v$ identifies the nodes at which the two
holes occur; $\Delta_1,\Delta_2$ are the output-free nested sequents placed
at those nodes.

\begin{figure}[ht]
\centering
\begin{mathpar}
\inferrule*[right=$p_{\Diamond}$]
 {\mathcal C\{\Delta_1\}_u\{A^{\circ},\Delta_2\}_v}
 {\mathcal C\{(\Diamond A)^{\circ},\Delta_1\}_u
              \{\Delta_2\}_v}
\quad u\prec v
\and
\inferrule*[right=$p_{\Box}$]
 {\mathcal C\{(\Box A)^{\bullet},\Delta_1\}_u
              \{A^{\bullet},\Delta_2\}_v}
 {\mathcal C\{(\Box A)^{\bullet},\Delta_1\}_u
              \{\Delta_2\}_v}
\quad u\prec v.
\end{mathpar}
\caption{Propagation along non-empty modal paths.}
\label{fig:nik4-prop}
\end{figure}

The second transitivity principle gives a small illustration.  Read upwards,
\[
\begin{aligned}
(\Diamond\Diamond p\to\Diamond p)^{\circ}
&\overset{\to^{\circ}}{\Longrightarrow}
  (\Diamond\Diamond p)^{\bullet},(\Diamond p)^{\circ}\\
&\overset{\Diamond^{\bullet}}{\Longrightarrow}
  (\Diamond p)^{\circ},[(\Diamond p)^{\bullet}]\\
&\overset{\Diamond^{\bullet}}{\Longrightarrow}
  (\Diamond p)^{\circ},[[p^{\bullet}]]\\
&\overset{p_{\Diamond}}{\Longrightarrow}
  [[p^{\bullet},p^{\circ}]].
\end{aligned}
\]
The last step uses the two-edge path from the root to the innermost node, and
that branch closes by $id$.

For a rule schema $r$, write $C\Longrightarrow_r(P_1,\ldots,P_k)$ for an
inference obtained by applying $r$, with conclusion $C$ and premises
$P_1,\ldots,P_k$; $k=0$ denotes a nullary inference.

The query for $A$ is represented by the one-node sequent $A^{\circ}$.  Read
upwards, the only rules which add a node are $\Diamond^{\bullet}$ and
$\Box^{\circ}$; in either case the fresh child carries the displayed
principal subformula.  No rule creates an empty modal child.

\subsection{Semantics of nested sequents}

Let $S$ be a full nested sequent with node set $V_S$, root $r_S$, modal edge
relation $E_S$, input set $\Gamma_S(v)$ at each node $v\in V_S$, and unique
output $A^{\circ}$ at the node $o_S$.  Its labelled-tree translation is the
labelled sequent
\[
 \{vRw:(v,w)\in E_S\},
 \{v:B\mid v\in V_S,\ B\in\Gamma_S(v)\}
 \ \vdash\ o_S:A.
\]
The sequent $S$ is valid when, for every birelational model and every assignment
$I:V_S\to W$ respecting all modal edges, truth of every input formula at its
assigned node implies truth of the output formula at $I(o_S)$.

Most rules are local.  Implication and box are not: falsifying them may
require shifting the whole modal assignment to a later intuitionistic world.

\begin{lemma}[Tree lifting]
\label{lem:tree-lifting}
Let $T$ be a finite modal tree and let $I$ respect its edges in a Fischer--Servi
frame.  If $I(v)\leq x$, then there is an edge-respecting assignment $I'$ such
that $I'(v)=x$ and $I(u)\leq I'(u)$ for every node $u$ of $T$.
\end{lemma}

\begin{proof}
Root the underlying undirected tree temporarily at $v$ and define $I'$ by
induction on the distance from $v$.  Put $I'(v)=x$.  Suppose the value of a
node has been fixed and consider an adjacent node not yet assigned.

If $s$ is assigned and $t$ is an original child of $s$, then
$I(s)\leq I'(s)$ and $I(s)RI(t)$.  Applying F2 to these two relations gives a
world $I'(t)$ such that
\[
 I'(s)RI'(t),
 \qquad I(t)\leq I'(t).
\]
If instead $t$ is assigned and $s$ is its original parent, then
$I(s)RI(t)\leq I'(t)$.  F1 gives a world $I'(s)$ with
\[
 I(s)\leq I'(s),
 \qquad I'(s)RI'(t).
\]
Every node has a unique predecessor in the temporary rooting at $v$, so it is
assigned exactly once.  Since $T$ is finite, every node is reached after
finitely many steps, and every original modal edge is respected in its
original direction.
\end{proof}

\begin{theorem}[Soundness]
\label{thm:nik4-sound}
Every sequent derivable in $\NIKfour$ is valid on transitive Fischer--Servi
frames.
\end{theorem}

\begin{proof}
For a rule with premises $P_1,\ldots,P_k$, it suffices to prove the
following contrapositive statement: every assignment which falsifies the
conclusion can be turned into an assignment which falsifies at least one
premise.  When $k=0$ there is no such assignment.  An $id$-conclusion would
require the same atom to be true as an input and false as the output at one
node, while a $\bot^{\bullet}$-conclusion would require $\bot$ to be true.
For the propositional rules other than implication, the claim follows
directly from the truth clauses.  In the branching cases, truth of an input
disjunction or falsity of an output conjunction selects the required
premise.

Let $I$ falsify a conclusion whose principal formula is
$(A\to B)^{\circ}$ at a node $u$, and put $x=I(u)$.  There is $y\geq x$ with
$y\Vdash A$ and $y\nVdash B$.  By \cref{lem:tree-lifting}, $I$ lifts to an
edge-respecting assignment $I'$ such that $I'(u)=y$ and
$I(v)\leq I'(v)$ at every inherited node.  Persistence keeps all inherited inputs true;
the added input $A^{\bullet}$ is true at $u$, whereas the new output
$B^{\circ}$ is false there.  Thus the premise is falsified.

Suppose next that $(A\to B)^{\bullet}$ is principal at $u$.  The conclusion
assignment makes $A\to B$ true at $I(u)$.  By reflexivity of $\leq$, either
$I(u)\nVdash A$ or $I(u)\Vdash B$: if $A$ is true, the implication clause at
$I(u)$ forces $B$.  In the first case delete the previous output and make
$A^{\circ}$ the output at $u$; all inputs, including
$(A\to B)^{\bullet}$, remain true, so the first premise is false.  In the second case use the second premise: its distinguished
$(A\to B)^{\bullet}$ occurrence is replaced by $B^{\bullet}$, while any
equal occurrence belonging to the side context is retained.  The previous false output is unchanged, so this premise is falsified.

If $(\Diamond A)^{\bullet}$ is principal at $u$, choose $y$ with
$I(u)Ry$ and $y\Vdash A$.  Extending $I$ by assigning $y$ to the fresh child
makes every premise input true and leaves the previous output false.  If
$(\Box A)^{\circ}$ is principal at $u$, its falsity gives worlds $y,z$ with
$I(u)\leq y$, $yRz$, and $z\nVdash A$.  Lift the original assignment at $u$ to
$y$, use persistence for every inherited input, and assign $z$ to the fresh child.
The premise output $A^{\circ}$ is then false.

Propagation deserves a separate line: strict descendancy in the nested tree
must be read as a genuine non-empty $R$-path.  If $u\prec v$, write the
unique modal path
as
\[
 u=v_0\to v_1\to\cdots\to v_n=v \qquad(n\geq1).
\]
An edge-respecting assignment gives
$I(v_0)R I(v_1)R\cdots R I(v_n)$; induction on $n$, using transitivity,
yields $I(u)R I(v)$.  Hence falsity of $\Diamond A$ at $I(u)$ implies
falsity of $A$ at $I(v)$, which falsifies the premise of $p_{\Diamond}$.
Likewise, if $\Box A$ is true at $I(u)$, reflexivity gives
$I(u)\leq I(u)$, and together with $I(u)R I(v)$ the box clause yields
$I(v)\Vdash A$.  Adding $A^{\bullet}$ at $v$ therefore falsifies the premise
of $p_{\Box}$.  These cases exhaust the rules.
\end{proof}

\subsection{Cut-free completeness}

Lyon parameterizes structural refinement by the Horn--Scott--Lemmon schemes
\[
 \varphi(n,k):=
 (\Diamond^{n}\Box A\to\Box^{k}A)\wedge
 (\Diamond^{k}A\to\Box^{n}\Diamond A),
 \qquad n,k\in\mathbb N.
\]
Here $\mathbb N$ includes $0$, and a zeroth modal power is read as
the identity: $\Diamond^{0}B=\Box^{0}B=B$.  For $n=0$ and $k=2$ this gives
\[
 \varphi(0,2)=
 (\Box A\to\Box\Box A)\wedge
 (\Diamond\Diamond A\to\Diamond A),
\]
exactly the two transitivity principles used here.  With $R^2=R\circ R$, the
frame condition associated with $\varphi(0,2)$ is $R^2\subseteq R$.  We
therefore set $\mathcal A_4=\{\varphi(0,2)\}$.

We take completeness from Lyon's labelled calculus.  Lyon's Theorem~4 applies to
$\mathcal A_4$-valid labelled sequents, whereas the decision argument below
uses the set-based nested calculus just displayed.  Theorem~4 supplies the
labelled proof; we translate it into Lyon's multiset nested presentation and
erase repetitions only afterwards, without appealing directly to the later
formula-level nested completeness statement.

Lyon's definition of a labelled sequent permits both the relational
multiset and the antecedent multiset to be empty.  Hence Theorem~4 applies
directly to the formula endsequent
\[
 \varnothing,\varnothing\vdash w:A.
\]
The only degenerate point arises one step later.  Lyon's literal definition of
a labelled tree sequent requires every formula label to occur in its relational
tree, whereas the endsequent above has no relational atom.  His translation
nevertheless contains an explicit case for an empty relational component, and
his formula-level nested completeness theorem presupposes the same convention
\cite[Definitions~14 and~16 and Theorem~6]{Lyon2021}.  We therefore read the
endsequent as the evident one-vertex rooted tree and verify the
labelled-to-nested passage below with that case included.

For $\mathcal A_4$, structural refinement uses the grammar
\[
 g(\mathcal A_4)=
 \{\Diamond\longrightarrow\Diamond\Diamond,
   \ \overline\Diamond\longrightarrow
      \overline\Diamond\,\overline\Diamond\}.
\]
Here $\overline\Diamond$ is the converse propagation symbol, used when a
relational edge is traversed against its displayed direction.  Since the
propagation rules begin with $\Diamond$, the relevant language is
\[
 L_{g(\mathcal A_4)}(\Diamond)
   =\{\Diamond^m:m\geq1\}.
\]
Their side condition is therefore a non-empty forward path.  The boundary
value $n=0$ does not license the empty word.

At the relational level, Lyon reads $wR^{0}u$ as $w=u$.  Thus the only
structural rule associated with $\mathcal A_4$ is
\[
\inferrule*[right=$S_{0,2}$]
 {\mathcal R,wRx,xRv,wRv,\Gamma\vdash z:C}
 {\mathcal R,wRx,xRv,\Gamma\vdash z:C}.
\]
If a propagation path does not use the shortcut $wRv$, this rule permutes over
the propagation inference immediately.  If it does, replace each occurrence
of that shortcut by the two-edge path $wRxRv$.  On path words this replaces
one occurrence of $\Diamond$ by $\Diamond\Diamond$, and the resulting word
still belongs to $L_{g(\mathcal A_4)}(\Diamond)$.  The same argument applies to
both propagation rules.  This is the $n=0$, $k=2$ instance of Lyon's
permutation lemmas and of the elimination procedure in
Theorem~3~\cite[Lemmas~1--2 and Theorem~3]{Lyon2021}.

For a labelled sequent
$\Lambda=\mathcal R,\Gamma\vdash u:C$, consider the directed graph whose
vertices are all labels occurring in $\mathcal R$, $\Gamma$, or $u:C$, with
an edge $x\to y$ whenever $xRy$ occurs in $\mathcal R$.  We say that
$\Lambda$ is \emph{tree-shaped with root $r$} when no relational atom is
repeated and this graph is a finite tree directed away from $r$.  The
one-vertex case $\mathcal R=\varnothing$ is included.  For such a sequent,
write $N_r(\Lambda)$ for the nested sequent obtained by placing at each vertex
$x$ one input occurrence $B^{\bullet}$ for every occurrence of $x:B$ in
$\Gamma$, placing the unique output $C^{\circ}$ at $u$, and nesting the
children according to this tree.

\begin{lemma}[From labelled to nested proofs]
\label{lem:labelled-to-nested}
Let $\mathcal D$ be a proof in Lyon's refined labelled calculus
$\mathsf{IK}(\mathcal A_4)\mathsf L$ whose conclusion $\Lambda$ is
tree-shaped with root $r$.  Then every sequent in $\mathcal D$ is tree-shaped
with the same root, and $\mathcal D$ can be transformed
effectively into a proof of $N_r(\Lambda)$ in the multiset presentation.  The
transformation preserves proof height.
\end{lemma}

\begin{proof}
By Lyon's Remark~2, replace every instance of $\Diamond r$ or $\Box l$ by
the height-preserving propagation instance along its displayed
edge~\cite[Remark~2]{Lyon2021}.  Read the normalized proof upwards.
Propositional and propagation rules leave the relational graph unchanged.
The rules $\Diamond l$ and $\Box r$ each add an atom $xRy$, where
$x$ already labels the principal formula in the conclusion and $y$ is fresh.
The new edge therefore attaches a single leaf and cannot duplicate an existing
edge.
No rule changes the root.  The refined calculus for $\mathcal A_4$ has no
seriality rule, and $S_{0,2}$ has been eliminated.  Hence tree shape and root
are preserved.  This is Lyon's Lemma~3 for non-degenerate relational
trees~\cite[Lemma~3]{Lyon2021}; the same inspection covers the one-vertex
case admitted here.

An induction on the height of $\mathcal D$ carries the translation through
the proof, with each labelled sequent read by $N_r$.  For the initial rules and the propositional rules other than
$\to l$, the relational graph is unchanged and every displayed formula remains
at the same label; applying $N_r$ therefore yields the corresponding nested
rule in \cref{fig:nik4-logical}.  For $\to l$, the first
labelled premise moves the output to the node carrying the principal
implication and places $A$ there.  Its translation is
$\mathcal C^{\downarrow}\{(A\to B)^{\bullet},A^{\circ}\}$.  The second premise
keeps the previous output and replaces the distinguished implication
occurrence by $B^{\bullet}$.  These are precisely the two premises of
$\to^{\bullet}$, including deletion of the previous output in the first.

The eigenlabel rules $\Diamond l$ and $\Box r$ translate to
$\Diamond^{\bullet}$ and $\Box^{\circ}$, with the fresh label becoming the
fresh child.  A propagation step from $x$ to $y$ is licensed by a word in
$L_{g(\mathcal A_4)}(\Diamond)$, hence by a non-empty forward path; in the
rooted tree this says that $y$ is a strict descendant of $x$.  It therefore
translates to $p_{\Diamond}$ or $p_{\Box}$.  Each labelled inference becomes
one nested inference, with root, output node, and output formula preserved.
Since the list of rule schemes is finite and each translation is
syntactic, the construction is effective and height-preserving.
\end{proof}

\begin{proposition}[Completeness of the multiset presentation]
\label{prop:multiset-complete}
Every formula valid on transitive Fischer--Servi frames has a cut-free proof
in the multiset presentation of the rules in
\cref{fig:nik4-logical,fig:nik4-prop}.
\end{proposition}

\begin{proof}
Let $A$ be valid on transitive Fischer--Servi frames.  The frame condition for
$\mathcal A_4$ is $R^2\subseteq R$, so the labelled sequent
\[
 \varnothing,\varnothing\vdash w:A
\]
is $\mathcal A_4$-valid.  Lyon's Theorem~4 yields a proof in the refined
labelled calculus
$\mathsf{IK}(\mathcal A_4)\mathsf L$~\cite[Theorem~4]{Lyon2021}.
Since cut is not a rule of $\mathsf{IK}(\mathcal A_4)\mathsf L$, the resulting
labelled proof is cut free.  Its conclusion is the one-vertex tree with root
$w$, and
\cref{lem:labelled-to-nested} translates the proof into a proof of
$N_w(\varnothing,\varnothing\vdash w:A)=A^{\circ}$ in the multiset
presentation.
\end{proof}

The two presentations differ only in their treatment of repeated input
formulae.  For a nested sequent $S$ in the multiset presentation, let
$S^{\flat}$ be obtained by erasing repeated input occurrences at every node;
the modal tree, the output, and sibling multiplicities are unchanged.

\begin{lemma}[Repetitions are inessential]
\label{lem:repetitions}
If $S$ has a derivation of height $h$ in the multiset presentation, then
$S^{\flat}$ has a derivation of height at most $h$ in $\NIKfour$.
Conversely, every inference of $\NIKfour$ is obtained by erasing
repetitions from a corresponding multiset inference.
\end{lemma}

\begin{proof}
Both directions come from the same elementary observation.  If
\[
 C\Longrightarrow_r(P_1,\ldots,P_k)
\]
is a multiset inference, then
\[
 C^{\flat}\Longrightarrow_r(P_1^{\flat},\ldots,P_k^{\flat})
\]
is an inference of the corresponding rule with set inputs.  Indeed, for input
multisets $M,N$, erasing repetitions commutes with filling
an input hole.  Writing $M\uplus N$ for multiset union,
\[
 (M\uplus N)^{\flat}=M^{\flat}\cup N^{\flat}.
\]
The modal tree, the unique output, freshness, and the chosen propagation path
are unchanged.  Distinguished input occurrences may be followed before
repetitions are erased.  If a principal occurrence has an equal copy in the
side context, the latter remains after the principal occurrence is removed or
replaced; if no such contextual copy was chosen, it does not.  These are
exactly the two applications allowed by the convention preceding the rules.  Coincident displayed subformulae and branching rules require no new
argument: the required multiset copies are distinguished before erasure, and
the same erased context occurs in every branch.  The first premise of
$\to^{\bullet}$ changes only the output, while the modal rules change only the
fresh child or the propagation path.  Thus the displayed local fact holds for
every rule.

Induction on derivation height lifts this local observation to whole proofs.
Apply it to the last inference and the induction hypothesis to each immediate
premise; the height does not increase.

Conversely, fix an inference with set inputs, including its chosen side
contexts.  Give every contextual input one multiset occurrence and add the
distinguished principal occurrences required by the rule.  When a principal
formula also belongs to the chosen side context, these are two occurrences;
when two displayed subformulae coincide, use the corresponding number of
distinguished copies.  Keep the output, tree, fresh child, and propagation
path unchanged.  This is a multiset inference, and erasing its repeated
inputs recovers the given inference with set inputs.
\end{proof}

This passage does not assume contraction admissibility.  Repeated input
occurrences are removed by the proof transformation just established;
$\NIKfour$ treats node inputs as sets, while retaining the multiplicities of
sibling subtrees.

\begin{theorem}[Cut-free completeness]
\label{thm:nik4-complete}
A formula $A$ is valid on transitive Fischer--Servi frames iff
$A^{\circ}$ is derivable in $\NIKfour$.
\end{theorem}

\begin{proof}
Soundness is \cref{thm:nik4-sound}.  Conversely, validity gives a cut-free proof in the multiset presentation by
\cref{prop:multiset-complete}.  Since $A^{\circ}$ contains no input
occurrences, $(A^{\circ})^{\flat}=A^{\circ}$; hence
\cref{lem:repetitions} turns that proof into a $\NIKfour$ proof of the same
endsequent.
\end{proof}

\subsection{Subformulas}

Fix the formula $A_0$ to be decided and let $\Sub(A_0)$ denote its set of
subformulae.  Put
\[
 \Phi=\Sub(A_0).
\]
The set $\Phi$ is finite and closed under taking subformulae.

\begin{lemma}[Subformula property]
\label{lem:subformula-property}
Every formula occurring in a cut-free $\NIKfour$-proof of $A_0^{\circ}$
belongs to $\Phi$.
\end{lemma}

\begin{proof}
Reading a cut-free rule upwards, every formula newly displayed in a premise is
a subformula of a principal formula in the conclusion; propagation merely moves
such a subformula.  The nullary rules add nothing to a branch: they only close a
nested sequent already reached.  Hence every formula in the proof belongs to
$\Phi$.  In particular, $\bot$ can occur only when it is already a subformula
of $A_0$.
\end{proof}

The same induction gives the form used below: if a full endsequent uses
only formulae from a finite set $\Phi$ closed under subformulae, then every
formula in any cut-free proof of it lies in $\Phi$.  Henceforth, all nested
sequents are full and use only formulae from this fixed $\Phi$; in particular,
all rule instances, predecessor operations, enumerations, and saturation
stages are taken inside that universe.

Input formulae at a node form a set, while sibling subtrees remain a finite
unordered family with multiplicity.  Each rule mentions finitely many nodes and
creates at most one fresh child.  On a fixed finite tree it is decidable whether
the rule applies; propagation requires only the additional test that one node
is a strict descendant of another.

\section{The homeomorphic order}

Cut elimination bounds the formulae, not the modal depth.  The trees must
therefore be compared without imposing an artificial depth parameter.  For a
nested sequent $X$, write $V_X$ and $E_X$ for its node and edge sets, $r_X$ for
its root, $o_X$ for its output node, and
$\Gamma_X(v)\subseteq\Phi$ for the set of input formulae at a node $v$.  Since
$\Phi$ is finite, only finitely many such sets and output entries are possible;
a full nested sequent has exactly one node carrying an output formula.

Homeomorphic embedding is suggested by the calculus itself.  Input formulae
may be added, and transitivity makes it harmless to stretch a modal edge into
a non-empty path.  The branching pattern must nevertheless be left intact:
two representing paths may meet only when the original edges share an
endpoint, and no image node may lie inside the path representing another edge.

\subsection{Homeomorphic embedding of nested sequents}

\begin{definition}[Homeomorphic embedding]\label{def:embedding}
For nested sequents $S,T$, write $S\preceq T$ if there is an injective map
$f:V_S\to V_T$ satisfying the following conditions.
\begin{enumerate}[label=(\roman*),leftmargin=2.4em]
\item $f(r_S)=r_T$;
\item $\Gamma_S(v)\subseteq\Gamma_T(f(v))$ for every $v\in V_S$, and
      $f(o_S)=o_T$ with the same output formula at these two nodes;
\item for every edge $e=(v,w)$ of $S$, $f(w)$ is a proper descendant of
      $f(v)$ in $T$; write $\pi_e$ for the resulting non-empty directed path
      from $f(v)$ to $f(w)$;
\item for distinct source edges $e,e'$, let $V(\pi)$ denote the vertex set
      of a path $\pi$ and require
      \[
        V(\pi_e)\cap V(\pi_{e'})=f(e\cap e'),
      \]
      where an edge is identified with its two endpoints and $f$ is
      applied pointwise to endpoint sets.  Writing $\operatorname{Int}$ for
      the set of internal vertices, we also require
      \[
        \operatorname{Int}(\pi_e)\cap f(V_S)=\varnothing.
      \]
\end{enumerate}
\end{definition}

The relation $\preceq$ is a quasi-order.  The identity map witnesses
reflexivity.  For transitivity, let $f:S\preceq T$ and $g:T\preceq U$, and
put $h=g\circ f$.  Root preservation is immediate under
composition.  The input condition follows from transitivity of $\subseteq$,
and the two output conditions give $h(o_S)=o_U$ with the same output
formula.  If an edge $e$ of $S$ is represented in $T$ by
\[
 f(v)=t_0\to t_1\to\cdots\to t_m=f(w),
\]
replace each edge $(t_{j-1},t_j)$ by its $g$-path in $U$ and concatenate the
resulting paths.  Consecutive paths meet at $g(t_j)$, non-consecutive ones are
disjoint, and the concatenation is a non-empty directed path from $h(v)$ to
$h(w)$.  For two distinct edges of $S$, their $f$-paths meet only at the
image of a common endpoint; clause~(iv) for $g$ therefore gives the same
intersection property after expansion.  Finally, an $h$-image node cannot be
internal to a concatenated path: this would make an $f$-image node internal to
an $f$-path, or a $g$-image node internal to a $g$-path.  Thus $h$ satisfies
clauses~(i)--(iv).

The image nodes together with the representing paths form a subtree $H$ of
$T$ with vertex set
\[
 f(V_S)\cup\bigcup_{e\in E_S}V(\pi_e).
\]
This also covers the one-node tree.  Clause~(iv) implies that every vertex of
$H$ outside $f(V_S)$ has degree $2$ in $H$.  Suppressing
precisely those vertices recovers a copy of the underlying unlabelled tree of
$S$, with $f(v)$ corresponding to $v$.  Conversely, a homeomorphism of the
underlying tree of $S$ onto a subtree of $T$ determines paths satisfying
clauses~(iii)--(iv).
Since the root is mapped to the root, orientation away from it makes these
paths directed descendant paths.  Thus clauses~(iii)--(iv) are precisely the
rooted homeomorphic-embedding conditions on the underlying trees, while
clause~(ii) records how the input formulae attached to a node may increase.

The relation is decidable.  For finite $S,T$, inspect the finitely many
injective maps $f:V_S\to V_T$ with $f(r_S)=r_T$ and test clauses~(ii)--(iv).

\subsection{Rules and proofs are monotone}

The calculus is monotone for this order.  Adding input formulae or
subdividing modal edges does not destroy an inference, and the lift must retain
enough information to be iterated through a proof.

To lift an inference along an embedding, we must identify in each premise the
nodes inherited from the conclusion.  For any inference $I$ with
conclusion $C$ and premises $P_i$, write $\operatorname{Old}(P_i)$ for the nodes of $P_i$ inherited from
$C$.  Since no rule identifies or duplicates a
conclusion node within one premise, there is a bijection
\[
 \iota_i^I:V_C\longrightarrow\operatorname{Old}(P_i),
\]
sending each conclusion node to its inherited copy.  The only possible node
outside this inherited part is the fresh child created by
$\Diamond^{\bullet}$ or $\Box^{\circ}$.

\begin{lemma}[Rules respect embedding]\label{lem:instance-lifting}
Suppose
\[
 C\Longrightarrow_r(P_1,\ldots,P_k)
\]
is an inference $I$ obtained by applying a rule $r$, possibly with $k=0$, and
$f:C\preceq D$.  Then one can effectively construct an inference
\[
 J:\quad D\Longrightarrow_r(Q_1,\ldots,Q_k)
\]
obtained by applying $r$, together with, for every $i$, an embedding
$f_i:P_i\preceq Q_i$ compatible with these canonical correspondences.  More
explicitly, starting from a conclusion node $x$, it makes no difference
whether we first take its inherited copy in $P_i$ and then apply $f_i$, or
first send $x$ to $f(x)$ and then take the inherited copy of $f(x)$ in
$Q_i$:
\[
 f_i(\iota_i^I(x))=\iota_i^J(f(x))
 \qquad(x\in V_C).
\]
\end{lemma}

\begin{proof}
Fix a multiset presentation of $I$.  At every conclusion node $x$ on which
the rule acts, its principal input formulae occur in $\Gamma_C(x)$ and hence in
$\Gamma_D(f(x))$.  Choose a multiset presentation of $D$ with one contextual occurrence for each
occurrence designated contextual in $I$.  Put every input formula of
$D$ not used by $I$ into the side context, and add the distinguished principal
occurrences required by the rule.  Thus a formula which is both
contextual and principal receives two occurrences.  Applying $r$ at the image
positions and then erasing repetitions gives an inference
\[
 J:\quad D\Longrightarrow_r(Q_1,\ldots,Q_k).
\]
For a modal introduction, the fresh child of $Q_i$ is attached to the image
of the principal node inherited from the conclusion.

Define $f_i$ on inherited nodes by
\[
 f_i(\iota_i^I(x))=\iota_i^J(f(x)) \qquad(x\in V_C),
\]
and, when $P_i$ has a fresh child, send it to the fresh child of $Q_i$.
The map is injective and sends root to root.  Let $A$ be an input formula at
an inherited node $\iota_i^I(x)$ of $P_i$.  Relative to the fixed multiset
inference, either $A$ is inherited from the side context of $C$ or it is
supplied by the rule.  In the first case $A\in\Gamma_C(x)$, hence
$A\in\Gamma_D(f(x))$, and the chosen side context retains it in $Q_i$.  In
the second case the lifted inference supplies $A$ at
$\iota_i^J(f(x))$.  The same alternative covers a fresh child, where every
displayed formula is supplied by the modal rule.  Therefore
\[
 \Gamma_{P_i}(v)\subseteq\Gamma_{Q_i}(f_i(v))
 \qquad(v\in V_{P_i}).
\]

For the output there are three possibilities.  It is inherited unchanged by
all input rules, by $p_{\Box}$, and by the second premise of
$\to^{\bullet}$; it is then preserved by $f$.  An output rule, the first
premise of $\to^{\bullet}$, or $p_{\Diamond}$ places the corresponding
output formula at corresponding inherited nodes.  Finally, $\Box^{\circ}$ places
it at the two corresponding fresh children.  Hence $f_i$ maps the unique
output node of $P_i$ to that of $Q_i$ and preserves its formula.

Every edge of $P_i$ between inherited nodes comes from an edge of $C$; use
for it the path by which $f$ represents that edge in $D$.  These paths retain
clause~(iv) of \cref{def:embedding}.  If $P_i$ has a fresh child and $x$ is the principal conclusion node, the
edge
from $\iota_i^I(x)$ to that child is represented by the corresponding fresh
edge in $Q_i$.  Its parent is $\iota_i^J(f(x))$, the inherited copy of an $f$-image node; clause~(iv) for $f$ says that $f(x)$ is not internal to any
path representing an inherited edge.  The fresh edge therefore meets every path representing an inherited edge only at its parent, if at all.  For a propagation rule, an active relation $u\prec v$ in $C$
is sent by $f$ to a non-empty path from $f(u)$ to $f(v)$, so the same side
condition holds in $D$.

When $k=0$, there are no premises to treat.  If $r=id$, clause~(ii) for
$f$ preserves the matching input atom at the output node; if
$r=\bot^{\bullet}$, it preserves $\bot$ at the active node.  Thus the same
nullary rule applies to $D$.  For $k>0$, the
preceding paragraphs verify all clauses of \cref{def:embedding} for each
$f_i$.  Every step uses only the finite data carried by $I$ and $f$, so the
inference $J$ and the embeddings $f_i$ are obtained effectively.
\end{proof}

\begin{lemma}[Weakening along embeddings]
\label{lem:weakening-embedding}
Given an embedding $f:S\preceq T$ and a proof of $S$ in $\NIKfour$, one can
effectively construct a proof of $T$ of no greater height.
\end{lemma}

\begin{proof}
Induct on the height of the given proof of $S$.  If its last inference is the
nullary inference $S\Longrightarrow_r()$, apply
\cref{lem:instance-lifting} to that inference and to $f$; the result is the
corresponding nullary inference with conclusion $T$.

Otherwise write the last inference as
\[
 S\Longrightarrow_r(P_1,\ldots,P_k).
\]
By \cref{lem:instance-lifting}, the embedding $f$ effectively yields an
inference $T\Longrightarrow_r(Q_1,\ldots,Q_k)$ and embeddings
$f_i:P_i\preceq Q_i$.  Applying the induction hypothesis to each $f_i$ and
the corresponding premise proof gives proofs of all $Q_i$ of no greater
height.  Reapplying $r$ proves $T$.
\end{proof}

Provability is therefore upward closed under $\preceq$.  A labelled form
of Kruskal's theorem supplies a stronger tree embedding and hence shows that
$\preceq$ is a well-quasi-order.

\subsection{Kruskal and finite bases}

\begin{lemma}[A decidable well-quasi-order]
\label{lem:wqo}
The relation $\preceq$ is a decidable well-quasi-order (wqo) on nested sequents.
\end{lemma}

\begin{proof}
We mark the root in the label, rather than pass to pointed trees, so that the
appeal to the labelled Tree Theorem is literal.
Let
\[
 Q_\Phi=\{\mathsf{root},\mathsf{other}\}
          \times\mathcal P(\Phi)
          \times(\Phi\cup\{\varnothing\}).
\]
The first coordinate is the root marker; the third is empty at every node
except the unique output node.  Declare a triple $(\rho,\Gamma,o)$ to be
below $(\rho',\Gamma',o')$ exactly when
\[
 \rho=\rho',\qquad \Gamma\subseteq\Gamma',\qquad o=o'.
\]
This is a finite quasi-order and hence a wqo.

Kruskal states the Tree Theorem for finite rooted structured trees whose
vertices are labelled in a wqo \cite[\S2, pp.~211--212]{Kruskal1960}.  His
structured trees carry a linear order on the edges leaving each vertex.  To
apply the theorem here, choose an arbitrary such order at every node of every
tree in an infinite sequence of nested sequents.  Kruskal then gives a pair
related by an order-preserving labelled monomorphism.  For these trees, each source edge is represented by an oriented path whose
interior contains no image node; paths representing distinct source edges meet
only in the image of a common endpoint and are otherwise vertex-disjoint.  After the auxiliary orders are
forgotten, the same map is therefore an embedding in the sense of
\cref{def:embedding}.  Thus the order used by Kruskal is stronger than
$\preceq$, which is all the wqo argument needs.
Nash--Williams's formulation makes explicit that the underlying unlabelled
relation is homeomorphism onto a subtree \cite[p.~833]{NashWilliams1963}.

The three coordinates of $Q_\Phi$ impose the remaining clauses.  Equality
of the root marker sends the sole node carrying $\mathsf{root}$ in the source
to the sole such node in the target.  Inclusion of the second coordinate is
the required inclusion of input formulae, and equality of the third preserves
both the output node and its formula.  Hence every infinite sequence of
nested sequents has a
$\preceq$-comparable pair.  The relation is therefore a wqo, and its
decidability was established above.
\end{proof}

Since $\preceq$ is a quasi-order rather than a partial order, minimality must
be understood modulo mutual embedding.  Write
\[
 S\equiv T \quad\Longleftrightarrow\quad
 S\preceq T\text{ and }T\preceq S,
 \qquad
 S\lhd T \quad\Longleftrightarrow\quad
 S\preceq T\text{ and }T\not\preceq S.
\]
Minimality below is always minimality modulo $\equiv$: an element $T$ of a set
$X$ is minimal when no $S\in X$ satisfies $S\lhd T$.  For a set $X$ of
nested sequents, put
\[
 \Up X=\{T:\exists S\in X\ (S\preceq T)\}.
\]
A set $B$ is a basis of an upward-closed set $U$ when $U=\Up B$.

\begin{proposition}[Finite bases and the ascending-chain condition]
\label{prop:wqo-consequences}
Every upward-closed set of nested sequents has a finite basis.  Moreover, every
increasing sequence of upward-closed sets of nested sequents stabilises.
\end{proposition}

\begin{proof}
The empty set has the empty basis.  Suppose $U\neq\varnothing$ and fix
$T_0\in U$.  If no $\lhd$-minimal member of $U$ lay below $T_0$, choose
$T_{n+1}\in U$ with $T_{n+1}\lhd T_n$ recursively.  This choice can be
continued: a minimal member below some $T_n$ would also lie below $T_0$.  For
$i<j$, transitivity gives $T_j\preceq T_i$, whereas $T_i\preceq T_j$ would
imply $T_i\preceq T_{i+1}$, contrary to $T_{i+1}\lhd T_i$.  The sequence
would therefore be bad.  Hence every member of $U$ lies above a minimal one.
Distinct minimal $\equiv$-classes are incomparable, and infinitely many such
classes would give an infinite bad sequence by choosing one representative
from each.  There are consequently only finitely many minimal classes.  One
representative from each forms a finite basis: its upward closure is contained
in $U$, while every member of $U$ lies above one of the representatives.

For the second assertion, suppose that an increasing chain
$(U_n)_{n\in\mathbb N}$ does not stabilise.  Choose recursively
\[
 n_0<n_1<n_2<\cdots
 \qquad\text{with}\qquad
 U_{n_i}\subsetneq U_{n_{i+1}},
\]
and choose $T_i\in U_{n_{i+1}}\setminus U_{n_i}$.  By
\cref{lem:wqo}, there are $i<j$ with $T_i\preceq T_j$.  Since
$n_{i+1}\leq n_j$, we have $T_i\in U_{n_{i+1}}\subseteq U_{n_j}$; upward
closure of $U_{n_j}$ then gives $T_j\in U_{n_j}$, contrary to the choice of
$T_j$.  Thus the chain stabilises.
\end{proof}

Kruskal yields finite bases and eventual stabilisation, but no algorithm for
computing predecessors or detecting the first stable stage.  Those two points
still have to be proved.

\section{Finite support and backward closure}

A backward inference may occur in an arbitrarily large context, although only
a few positions are used by the rule and by the images of the chosen basis
elements.  The finite-support argument removes everything else and reduces the
computation of predecessor bases to a finite search.

\subsection{Finite support}

Recall that $C\Longrightarrow_r(P_1,\ldots,P_k)$ denotes an inference
obtained by applying $r$, with conclusion $C$ and the displayed premises.
Nullary inferences are initial
sequents.  For an upward-closed set $U$, define
\[
 \Pre_r(U)=
 \{C:\exists P_1,\ldots,P_k\in U\quad
       C\Longrightarrow_r(P_1,\ldots,P_k)\},
 \qquad
 \Pre(U)=\bigcup_r\Pre_r(U),
\]
where the union ranges over the finitely many rule schemas.  For $k=0$ the
premise condition is empty.  For a finite nested sequent $X$, write $|X|:=|V_X|$.

By \cref{lem:instance-lifting}, $\Pre_r(U)$ is upward closed whenever
$U$ is.  Indeed, lift an inference with conclusion $C$ along
$C\preceq D$; each premise embeds into the corresponding lifted premise,
which therefore belongs to $U$.  Nullary rules require no separate argument.

Most of the surrounding context is irrelevant to a fixed backward inference,
but the nodes that determine it must remain.  We call these positions
\emph{active}.
For $id$, $\bot^{\bullet}$, and every propositional rule
except $\to^{\bullet}$, the sole active position is the node carrying the
displayed principal formulae.  For $\to^{\bullet}$ the principal node and the
conclusion-output node are active, even when they coincide.  A modal introduction acts at one conclusion node, its principal node; the
fresh child belongs only to the premise and will be regenerated.  A
propagation rule acts at two conclusion nodes.  Let $a_r$ be the maximum
number of distinct conclusion nodes at which an $r$-instance acts.  Thus $a_r=1$ for the
nullary rules, the local propositional rules, and the modal introductions,
while $a_r=2$ for $\to^{\bullet}$ and the propagation rules.

For each inference, fix a multiset presentation whenever a displayed formula
already occurs in the side context.  It records which occurrence the inference uses.  In a
premise, formulae at nodes inherited from the conclusion are either already present
there or
supplied by the rule; only the inherited ones must be retained when the
conclusion context is restricted.

\begin{lemma}[Restricting an inference]\label{lem:rule-restriction}
Consider an inference $I$ of the form
\[
 C\Longrightarrow_r(P_1,\ldots,P_k).
\]
Let $M\subseteq V_C$ contain the root, the conclusion output node, and every conclusion node at which $I$ acts.  At nodes of $M$, choose some input formulae from the conclusion context to
retain, including the input formulae which make a nullary conclusion initial
and
all principal input formulae required in the conclusion.  Form $D$ as follows:
\begin{enumerate}[label=(\roman*),leftmargin=2.4em]
\item take the union of the root-to-$v$ paths in $C$ for $v\in M$;
\item remove all branches outside this rooted hull and all unchosen input
      formulae, but retain the conclusion output formula;
\item simultaneously suppress every maximal chain of unmarked vertices whose
      undirected degree in that rooted hull is $2$ (equivalently, iterate the
      suppression until no such vertex remains).
\end{enumerate}
Then $D\preceq C$, and the same rule can be applied to give an inference
\[
 D\Longrightarrow_r(Q_1,\ldots,Q_k).
\]
Every chosen premise formula coming from the conclusion context remains at
its corresponding node; the rule supplies its other displayed formulae again
and regenerates any fresh modal child.
\end{lemma}

\begin{proof}
Let $K$ be the rooted hull in clause~(i), and identify every surviving vertex
of $D$ with the corresponding vertex of $C$.  Each edge of $D$ is represented
by the unique non-empty segment of $K$ between its endpoints.  Distinct such
segments have disjoint interiors and no surviving vertex lies in an interior.
The inclusion of surviving vertices sends the marked root of $D$ to
$r_C$ and therefore witnesses $D\preceq C$: only input formulae have been
deleted, while the marked output node and its formula are unchanged.

To reconstruct the inference, restrict the fixed multiset side context of $I$
to the retained occurrences, keep the same distinguished principal occurrences
at the surviving active nodes, and erase repetitions after the rule has been
applied.  The verification is as follows.
For $id$ and $\bot^{\bullet}$, the active node and the formulae which make the
conclusion initial were retained.  A propositional rule other than
$\to^{\bullet}$ refers only to its principal node and displayed principal
formulae, all of which survive.  For $\to^{\bullet}$, both the principal node
and the node carrying the conclusion output survive; the first premise removes
that output and puts $A^{\circ}$ at the principal node, whereas the second
retains it and puts $B^{\bullet}$ at the principal node.  A modal introduction
uses its surviving principal node and regenerates its fresh child.

The only rules whose side condition relates two inherited nodes are
$p_{\Diamond}$ and $p_{\Box}$.  Their endpoints belong to $M$.  If they are
$u\prec v$ in $C$, the root-to-$v$ path in $K$ contains $u$ and the whole
segment from $u$ to $v$.  Since $u\neq v$ and both survive, suppressing
unmarked degree-$2$ vertices leaves a non-empty path from $u$ to $v$ in $D$.
Thus strict descendancy is preserved.  Every chosen inherited input remains
at the corresponding inherited node, and every other premise formula is
supplied again by the
rule.  This yields the required inference with conclusion $D$.
\end{proof}

At most $a_r$ conclusion nodes are used by the rule.  The root and output may
have to be kept independently, and each node of a basis element contributes
at most one origin in the conclusion.  Hence the marked set will have size at
most
\[
 a_r+2+\sum_{i=1}^{k}|B_i|.
\]
Its rooted hull may also contain unmarked branch points.  The following
small tree estimate accounts for them.

\begin{lemma}[Size of the reduced hull]
\label{lem:reduced-hull}
Let $K$ be the union of the root-to-$v$ paths in a finite rooted tree, where
$v$ ranges over a non-empty set $M$ of marked vertices containing the root.
Suppress every maximal chain consisting of unmarked vertices of undirected
degree $2$ in $K$.  The resulting tree has at most $2|M|-1$ vertices.
\end{lemma}

\begin{proof}
Every leaf of $K$ is marked, since it is the endpoint of one of the paths used
to form the hull.  After suppression, every surviving unmarked vertex has
undirected degree at least $3$.  Put $m=|M|$ and let $b$ be the number of
surviving unmarked vertices.  If $m=1$, the hull is just its marked root and
$b=0$.  If $m\geq2$, every marked vertex has degree at least $1$, and the
degree sum gives
\[
 2(m+b-1)\geq m+3b.
\]
Thus $b\leq m-2$, and in either case the total number of vertices is at most
$2m-1$.
\end{proof}

\Needspace{13\baselineskip}
\begin{lemma}[Finite support]
\label{lem:finite-support}
Suppose an application of $r$ has conclusion $C$ and premises
$P_1,\ldots,P_k$, and let $B_1,\ldots,B_k$ be nested sequents equipped with
fixed embeddings $e_i:B_i\preceq P_i$.  For $k=0$ the sum below is empty.  Then
there is another application of $r$ with conclusion $C'\preceq C$ and
premises $P'_i$ such that
$B_i\preceq P'_i$ for every $i$ and
\[
 |C'|\leq 2\left(a_r+2+\sum_{i=1}^{k}|B_i|\right)-1.
\]
\end{lemma}

\begin{figure}[H]
\centering
\setlength{\unitlength}{0.68mm}
\begin{picture}(220,58)
\put(25,48){\circle*{4}}
\put(25,47){\line(-1,-1){10}}
\put(25,47){\line(0,-1){10}}
\put(25,47){\line(1,-1){10}}
\put(15,35){\circle{4}}
\put(22.5,32.5){\framebox(5,5){}}
\put(35,35){\circle{4}}
\put(15,33){\line(-1,-1){10}}
\put(15,33){\line(0,-1){10}}
\put(25,32){\line(0,-1){10}}
\put(35,33){\line(0,-1){10}}
\put(35,33){\line(1,-1){10}}
\put(5,21){\circle*{4}}
\put(15,21){\circle{4}}
\put(25,21){\circle{4}}
\put(35,21){\circle*{4}}
\put(45,21){\circle{4}}
\put(23,8){$C$}
\put(53,34){\vector(1,0){22}}
\put(58,43){\makebox(15,0){\small rooted}}
\put(58,39){\makebox(15,0){\small hull}}
\put(100,48){\circle*{4}}
\put(100,47){\line(-1,-1){10}}
\put(100,47){\line(0,-1){10}}
\put(100,47){\line(1,-1){10}}
\put(90,35){\circle{4}}
\put(97.5,32.5){\framebox(5,5){}}
\put(110,35){\circle{4}}
\put(90,33){\line(0,-1){10}}
\put(100,32){\line(0,-1){10}}
\put(110,33){\line(0,-1){10}}
\put(90,21){\circle*{4}}
\put(100,21){\circle{4}}
\put(110,21){\circle*{4}}
\put(93,8){\small hull}
\put(128,34){\vector(1,0){22}}
\put(130,43){\makebox(20,0){\small suppress}}
\put(128,39){\makebox(24,0){\small unary chains}}
\put(185,48){\circle*{4}}
\put(185,47){\line(-1,-1){10}}
\put(185,47){\line(0,-1){10}}
\put(185,47){\line(1,-1){10}}
\put(175,35){\circle*{4}}
\put(182.5,32.5){\framebox(5,5){}}
\put(195,35){\circle*{4}}
\put(182,18){$C'$}
\end{picture}
\caption{Marked nodes, their rooted hull, and the tree left after unmarked
unary stretches are suppressed.  Filled and boxed nodes are marked; the
box singles out an active position, while the filled nodes represent the
remaining marked vertices.}
\label{fig:finite-support}
\end{figure}
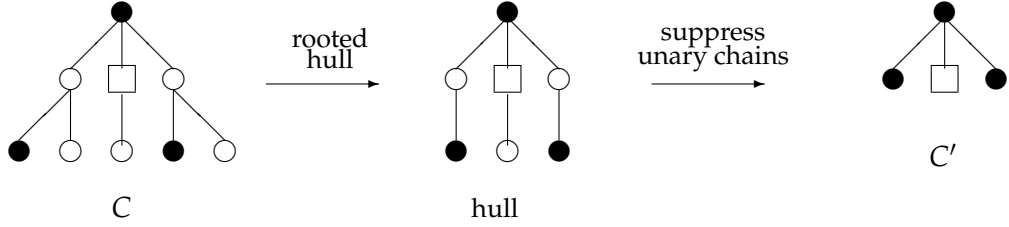

\begin{proof}
We begin with the marked nodes; the tree estimate will be used only after the
required formulae and paths have been shown to survive.  Let $I$ be the given
inference.  Some nodes in the images of the $B_i$ lie at premise nodes inherited from the
conclusion.  To mark their corresponding conclusion nodes, invert the
canonical correspondence in each premise:
\[
 \operatorname{org}_i^I=(\iota_i^I)^{-1}:
 \operatorname{Old}(P_i)\longrightarrow V_C.
\]
This is well defined because a rule neither identifies nor duplicates a
conclusion node within one premise.  Let $M\subseteq V_C$ consist of the root,
the output node, every active conclusion node of $I$, and every node
$\operatorname{org}_i^I(e_i(v))$ with $v\in V_{B_i}$ and
$e_i(v)\in\operatorname{Old}(P_i)$.  Repetitions are ignored.  There are at
most $a_r$ active nodes, and each node of $B_i$ contributes at most one further
node of the conclusion.  Hence
\[
 |M|\leq a_r+2+\sum_i|B_i|.
\]

The marked vertices determine the shape of the smaller conclusion, but not
yet the formulae that must survive there.  The formulae to retain are those
needed to reconstruct the premises.  Keep the conclusion output at its marked node.  Among the inputs, keep every principal
formula required in the conclusion of $I$ and every formula needed to make a
nullary conclusion initial.  Fix $i$, $v\in V_{B_i}$, and $A\in\Gamma_{B_i}(v)$.  If $e_i(v)$ lies at an inherited node, choose an occurrence of $A$ at
$e_i(v)$ in the fixed multiset presentation of $P_i$.  When that occurrence
is inherited from the conclusion side context, retain it at
$\operatorname{org}_i^I(e_i(v))$; when it is supplied by the rule, retain
nothing, since the restricted inference supplies it again.  If $e_i(v)$ is
fresh, then $P_i$ is a premise of $\Diamond^{\bullet}$: its fresh child has
no side context, and its sole input is supplied again by the restricted
inference.  (The fresh child introduced by $\Box^{\circ}$ carries the output
and no inputs.)  No other input formulae are needed when the premises are
reconstructed below.

Let $K$ be the union in $C$ of the paths from $r_C$ to the nodes of $M$.
Delete every branch outside $K$ and every unchosen input formula.  Suppress
the maximal unmarked degree-$2$ chains exactly as in
\cref{lem:rule-restriction}, and call the resulting nested sequent $C'$.  All marked vertices survive.  Interpreting every new edge
as the corresponding segment of $K$ gives a canonical embedding
$C'\preceq C$.  By \cref{lem:rule-restriction}, the restricted context
supports an application of $r$
\[
 I':\quad C'\Longrightarrow_r(P'_1,\ldots,P'_k).
\]

The same marked data recover the chosen embeddings in the restricted premises.  Define $e'_i:V_{B_i}\to V_{P'_i}$ by
\[
 e'_i(v)=
 \begin{cases}
 \iota_i^{I'}\bigl(\operatorname{org}_i^I(e_i(v))\bigr),
   & e_i(v)\in\operatorname{Old}(P_i),\\[2mm]
 \text{the fresh child of }P'_i,
   & e_i(v)\text{ is fresh}.
 \end{cases}
\]
The first line is defined because the corresponding conclusion node was marked.  The second line can
apply only to a leaf of $B_i$, since a fresh child of a premise is a leaf.
The map is injective: images at inherited nodes retain their distinct origins and the fresh
child lies outside the inherited part.  It also preserves the root.  Indeed, an
embedding sends the root of $B_i$ to the root of $P_i$, which is inherited from the conclusion, and the
first clause sends its origin to the root of $P'_i$.

The chosen occurrences ensure that
\[
 \Gamma_{B_i}(v)\subseteq\Gamma_{P'_i}(e'_i(v))
\]
for every $v$.  The output condition is checked in the same three cases as in
\cref{lem:instance-lifting}.  An inherited output was retained with its node;
an output produced at an inherited node is produced again at the
corresponding marked node; and $\Box^{\circ}$ produces it at the regenerated fresh child.
Thus $e'_i$ preserves the unique output and its formula.

Let $(v,w)$ be an edge of $B_i$.  If both $e_i(v)$ and $e_i(w)$ lie at inherited nodes,
then $e_i(v)$ is an ancestor of $e_i(w)$ in the copy inherited from $C$.  Both corresponding conclusion nodes are marked, so the root-to-target path contained in $K$ contains the
whole segment between them.  Its endpoints are distinct marked
vertices; suppression may shorten the segment but cannot make it empty.  If $e_i(w)$ is fresh, the part of the representing path inside the inherited tree ends at the modal
principal node.  That node is active and marked, and the fresh edge in $P'_i$
completes the path.  No edge of $B_i$ can start at a fresh image, because the
fresh child is a leaf.

For clause~(iv) of \cref{def:embedding}, every branch point
of the union of the representing paths lies in the rooted hull $K$.  Such a
point is either an image node, hence marked, or has degree at least $3$ in
$K$; in either case it survives.  Simultaneous suppression is performed only
inside the maximal unmarked degree-$2$ chains between surviving vertices.  It
therefore preserves the incidence graph of all relevant routes: it neither
identifies surviving vertices nor joins two previously distinct segments.
Consequently distinct edge paths still meet only at the image of a common
endpoint, and no image node becomes internal to another edge path.  Hence
$e'_i:B_i\preceq P'_i$.

The tree underlying $C'$ is obtained from $K$ by the suppression just
described.  By \cref{lem:reduced-hull},
\[
 |C'|\leq2|M|-1\leq
 2\left(a_r+2+\sum_i|B_i|\right)-1.
\]
\end{proof}

\begin{corollary}[Bounded minimal predecessors]
\label{cor:bounded-predecessors}
Let $U=\Up B$ with $B$ finite, and let $r$ have $k$ premises.  If $k>0$ and
$B=\varnothing$, then $\Pre_r(U)=\varnothing$.  If $k=0$, every minimal
member of $\Pre_r(U)$ has at most $2(a_r+2)-1$ nodes.  If
$k>0$ and $B\neq\varnothing$, the premises of a single inference need not lie
above the same member of $B$: for the $i$th premise one may have to choose a
separate $B_i\in B$.  The finite-support bound therefore depends first on a
tuple $(B_1,\ldots,B_k)\in B^k$.  Taking the maximum over the finitely many
such tuples gives the uniform bound
\[
 \max_{(B_1,\ldots,B_k)\in B^k}
 2\left(a_r+2+\sum_{i=1}^{k}|B_i|\right)-1.
\]
\end{corollary}

\begin{proof}
Let
\[
 C\Longrightarrow_r(P_1,\ldots,P_k)
\]
witness that the minimal nested sequent $C$ belongs to $\Pre_r(U)$.  When
$k>0$, choose $B_i\in B$ with $B_i\preceq P_i$; for $k=0$ use the empty
tuple.  The finite-support lemma gives another inference
\[
 C'\Longrightarrow_r(P'_1,\ldots,P'_k)
\]
with $C'\preceq C$, within the displayed bound, and
$B_i\preceq P'_i$.  Since $U=\Up B$, every $P'_i$ belongs to $U$; hence
$C'\in\Pre_r(U)$.  Minimality of $C$ modulo $\equiv$ now gives
$C\preceq C'$, for otherwise $C'\lhd C$.  The embeddings in both directions
are injections between finite node sets, so $|C|=|C'|$, and the same bound
holds for $C$.
\end{proof}

\subsection{Computing predecessors}

A node bound leaves only finitely many nested sequents and rule applications
to inspect.

\begin{proposition}
\label{prop:bounded-sequents}
For fixed finite $\Phi$ and $n\in\mathbb N$, nested sequents with at most $n$
nodes are effectively enumerable.  For each such nested sequent and each rule schema
$r$, the applications of $r$ having that nested sequent as conclusion are
effectively enumerable.  The relation $\preceq$ is decidable on finite
nested sequents, and whenever $S\preceq T$, a witnessing embedding can be found
effectively.
\end{proposition}

\begin{proof}
Number the root by $0$ and record the parent of every other node, requiring
that repeated passage to parents reaches $0$.  Record also an input subset of
$\Phi$ at each node and the unique output node together with its output
formula.  For a fixed bound there are only finitely many such descriptions.  Some
of them describe isomorphic unordered trees, but the repetition is harmless.

For a fixed conclusion with at most $n$ nodes, the part of every premise
inherited from the conclusion has the same modal tree; a modal introduction
may add one fresh leaf.
Thus every premise of an application under consideration has at most $n+1$
nodes.  The formula variables of a rule range over the fixed set $\Phi$;
once a principal formula is chosen, its displayed immediate subformulae are
fixed.  An application is determined by finitely many choices: a rule schema,
the conclusion nodes at which it acts, its principal formulae, and, for each displayed input
formula, whether an equal occurrence is also retained in the side context; a
propagation rule additionally requires its two endpoints.  There are therefore finitely many descriptions to inspect.  A description
is accepted when its premise and conclusion data satisfy the displayed rule
equation.  Equality of contexts, the chosen principal occurrences, freshness,
and strict descendancy are all decidable from the finite data.  Thus the finite
enumeration is filtered effectively; computability is not being inferred from
finiteness alone.

Deciding $S\preceq T$ is also a finite search: enumerate the injections from
$V_S$ to $V_T$.  For each injection, check the root, input, and output conditions of
\cref{def:embedding}.  The parent data determine the unique paths between
comparable nodes, so the remaining path-intersection and internal-vertex
conditions are finite tests as well.  If an injection passes all of them,
it is a witnessing embedding; hence the same finite search also produces a
witness whenever one exists.
\end{proof}

\begin{lemma}[Computing a basis of predecessors]
\label{lem:predecessor-basis}
Let $U=\Up B$ with $B$ finite.  For every rule schema $r$, one can compute a
finite basis of $\Pre_r(U)$.
\end{lemma}

\begin{proof}
If $r$ has $k>0$ premises and $B=\varnothing$, then $\Up B=\varnothing$
and $\Pre_r(\Up B)=\varnothing$.  If $k=0$, let $\vec B$ be the empty
tuple.  Otherwise, each premise in $\Up B$ lies above some member of $B$.
We may therefore separate the predecessors according to the tuple of
basis elements chosen below their premises, and fix
$\vec B=(B_1,\ldots,B_k)\in B^k$.  In either case, let $S_{\vec B}$
consist of the conclusions $C$ for which $r$ can be applied as
\[
 C\Longrightarrow_r(P_1,\ldots,P_k)
 \qquad\text{with}\qquad B_i\preceq P_i\quad(1\leq i\leq k).
\]
The finite-support lemma bounds representatives of the minimal classes of
$S_{\vec B}$.  For every $C\in S_{\vec B}$ there is a
$D\in S_{\vec B}$ such that $D\preceq C$ and
\[
 |D|\leq2\left(a_r+2+\sum_i|B_i|\right)-1.
\]
If $C$ is minimal modulo $\equiv$, then $C\preceq D$ as well, so its
$\equiv$-class already has a representative satisfying this bound.  We
therefore consider the finite set
\[
 F_{\vec B}=\left\{C\in S_{\vec B}:
 |C|\leq2\left(a_r+2+\sum_i|B_i|\right)-1\right\}.
\]

By \cref{prop:bounded-sequents}, we can enumerate all nested sequents
within the bound and all applications of $r$ having one of them as conclusion.
We retain precisely those applications for which $B_i\preceq P_i$ for every $i$;
these tests are decidable.  Hence $F_{\vec B}$ is effectively computable.
Let $M_{\vec B}$ contain one representative from each $\equiv$-class of
those $C\in F_{\vec B}$ for which no $D\in F_{\vec B}$ satisfies
$D\lhd C$.  Pairwise embedding tests compute this set.

Every $C\in S_{\vec B}$ lies above some element of $F_{\vec B}$.  Repeated
strict descent within the finite set $F_{\vec B}$ must terminate, and its
last element belongs to one of the classes represented in $M_{\vec B}$.
Hence some member of $M_{\vec B}$ embeds into $C$.
  Conversely, $S_{\vec B}$ is
upward closed: lifting the inference along $C\preceq C'$ gives premises above
the original premises, and therefore still above the fixed $B_i$.  Since every
member of $M_{\vec B}$ lies in $S_{\vec B}$, we obtain
\[
 S_{\vec B}=\Up M_{\vec B}.
\]

A conclusion belongs to $\Pre_r(\Up B)$ exactly when it belongs to some
$S_{\vec B}$: choose below each premise a basis member $B_i\in B$, and
conversely use upward closure of $\Up B$.  Thus the union of the finitely many
sets $M_{\vec B}$ is a basis of $\Pre_r(\Up B)$; minimizing that union modulo
$\equiv$ gives the desired basis.  For $k=0$ there is one empty tuple, no
premise test, and the bound is $2(a_r+2)-1$.
\end{proof}

\subsection{Proof height and stabilisation}

The nullary rules supply the first basis: apply
\cref{lem:predecessor-basis} with $B=\varnothing$ to each such rule whose
conclusions use only formulae from $\Phi$: the
$id$-conclusions for atoms $p\in\Phi$, and $\bot^{\bullet}$-conclusions
only when
$\bot\in\Phi$.  Since $k=0$, the premise condition is vacuous.  Minimize the union of these bases modulo
$\equiv$ and call the result $B_0$.  
For both schemas $a_r=1$, so every minimal contextual initial sequent has at
most
\[
 2(1+2)-1=5
\]
nodes.
  By \cref{lem:instance-lifting}, contextual initial sequents are upward
closed, so $\Up B_0$ is exactly their set.  These stages simply record proof height.  Take an initial sequent to have height $0$, and a
non-nullary inference to have height one plus the maximum height of its premise
proofs.  For an upward-closed set $U$, put
\[
 \mathcal F(U)=U\cup\Pre(U).
\]
Starting from the contextual initial sequents, set
\[
 U_0=\Up B_0,
 \qquad
 U_{n+1}=\mathcal F(U_n).
\]
By \cref{lem:predecessor-basis}, a finite basis of each $U_{n+1}$ is computed from
a finite basis of $U_n$ by taking the union of the existing basis with the finitely
many bases for the predecessor sets and minimizing modulo $\equiv$.

\begin{lemma}[Proof height and the stages $U_n$]
\label{lem:height}
For every $n$ and every nested sequent $C$,
\[
 C\in U_n
 \quad\Longleftrightarrow\quad
 C\text{ has a }\NIKfour\text{ proof of height at most }n.
\]
Consequently $\bigcup_n U_n$ is exactly the set of provable nested sequents.
\end{lemma}

\begin{proof}
For $n=0$, the set $U_0$ consists exactly of the contextual initial sequents,
which are precisely the nested sequents with proofs of height $0$.  Suppose the
equivalence holds at stage $n$.  If $C\in U_{n+1}$, then either $C\in U_n$ or
$C\in\Pre_r(U_n)$ for some rule $r$.  In the first case the induction
hypothesis gives a proof of height at most $n$.  In the second, every premise
of one application of $r$ has a proof of height at most $n$; appending that
inference gives a proof of $C$ of height at most $n+1$.

Conversely, let $C$ have a proof of height at most $n+1$.  A proof of height
$0$ is initial, hence its conclusion lies in $U_0\subseteq U_{n+1}$.  If the
proof has positive height, every premise of its last inference has height at
most $n$ and therefore belongs to $U_n$ by the induction hypothesis.  Thus
$C\in\Pre(U_n)\subseteq U_{n+1}$.
\end{proof}

\begin{theorem}[Stabilisation]
\label{thm:stabilisation}
One can find an $N$ such that
\[
 U_N=U_{N+1}=\cdots,
\]
and the stable set is exactly the set of provable nested sequents.
\end{theorem}

\begin{proof}
The set $U_0$ is upward closed, and \cref{lem:instance-lifting} shows
inductively that every $U_n$ is upward closed.  The sequence is increasing, so
\cref{prop:wqo-consequences} gives an $N$ with $U_N=U_{N+1}$.  Since
$U_{N+1}=\mathcal F(U_N)$, this equality says that $U_N$ is a fixed point of
$\mathcal F$.  Every later stage is therefore the same set.

The first stable index is effectively recognisable.  To compare two finitely
based upward-closed sets, it is enough to compare their generators.  The
orientation is forced by upward closure: every generator of the left-hand
side must already lie above some generator of the right-hand side.  Thus, for
finite sets $B,C$, not necessarily minimized,
\[
 \Up B\subseteq\Up C
 \quad\Longleftrightarrow\quad
 \forall b\in B\;\exists c\in C\ (c\preceq b).
\]
For the forward direction, $b\in\Up B$ and the inclusion gives
$b\in\Up C$.  Conversely, if $x\in\Up B$, choose $b\in B$ with
$b\preceq x$ and then $c\in C$ with $c\preceq b$; transitivity gives
$x\in\Up C$.  Since $\preceq$ is decidable, mutual inclusion of finitely based
upward-closed sets is decidable.  Successive upward closures can therefore be
compared effectively, and the first equality is detected when it occurs.

By \cref{lem:height}, $\bigcup_n U_n$ is exactly the set of nested sequents
having finite $\NIKfour$ proofs.  Once the chain stabilises, this union is
$U_N$.
\end{proof}

\begin{theorem}[Decidability of $\IKfour$]
\label{thm:main}
Theoremhood in Simpson's intuitionistic modal logic $\IKfour$, equivalently
validity on transitive Fischer--Servi frames, is decidable.
\end{theorem}

\begin{proof}
Given $A_0$, perform the following finite saturation:
\begin{enumerate}[label=(\arabic*),leftmargin=2.4em]
\item Compute the finite set $\Phi=\Sub(A_0)$.
\item Compute the finite initial basis $B_0$ from the nullary rules.
\item Given a finite basis $B_n$ for $U_n$, compute finite bases of
      $\Pre_r(U_n)$ for every rule schema $r$, take their union with $B_n$,
      and minimise it modulo $\equiv$ using the pairwise embedding tests
      of \cref{prop:bounded-sequents} to obtain a basis $B_{n+1}$.
\item Decide whether $\Up B_n=\Up B_{n+1}$ by the finite-basis inclusion
      test of \cref{thm:stabilisation}.  If not, return to step~(3).
\item At the first equality, put $B_{\infty}=B_n$ and form the one-node
      nested sequent $C(A_0)$ whose output is $A_0$.
\item Accept exactly when $D\preceq C(A_0)$ for some
      $D\in B_{\infty}$.
\end{enumerate}
Every operation in steps~(1)--(3) is effective by
\cref{prop:bounded-sequents,lem:predecessor-basis}; the test in step~(4) is
decidable; and \cref{prop:wqo-consequences} guarantees that step~(4) eventually
succeeds.  Thus the procedure is total.  At the first stable stage, $U_{\infty}=\Up B_{\infty}$ is exactly the set of provable
nested sequents over $\Phi$, so \cref{lem:height} yields
\[
 A_0^{\circ}\text{ is derivable}
 \quad\Longleftrightarrow\quad
 \exists D\in B_{\infty}\ (D\preceq C(A_0)).
\]
Because $C(A_0)$ has one node and no inputs, any
$D\preceq C(A_0)$ must itself be the one-node sequent with output $A_0$ and no
inputs.  Thus the last basis test is exactly the original theoremhood query,
not a weakening of it.  The test is decidable because $B_{\infty}$ is finite
and $\preceq$ is decidable by \cref{lem:wqo}.  Soundness
\cref{thm:nik4-sound} and cut-free completeness
\cref{thm:nik4-complete} identify $\NIKfour$-derivability with validity on
transitive Fischer--Servi frames, hence with $\IKfour$.
\end{proof}

Nothing in the construction depended on $\Phi$ being the subformula set of
a single formula; finiteness and closure under subformulae were enough.  The
stable basis therefore contains a little more information than decidability
alone records.

\begin{corollary}[Finite generators for cut-free provability]
\label{cor:finite-generation}
Let $\Phi$ be any finite set of formulae closed under subformulae.  One can
effectively compute a finite family of cut-free $\NIKfour$-proofs
$\pi_1,\ldots,\pi_m$, with respective conclusions $D_1,\ldots,D_m$, and a
number $N$ such that, for every nested sequent $C$ whose formulae lie in
$\Phi$,
\[
 C\text{ is provable}
 \quad\Longleftrightarrow\quad
 \exists j\ (1\leq j\leq m\ \&\ D_j\preceq C).
\]
Whenever $C$ is provable, a proof of $C$ of height at most $N$ can be
constructed effectively from a suitable $\pi_j$ by weakening along an
embedding.
\end{corollary}

\begin{proof}
Iterate backward application of the rules over the chosen $\Phi$ until the
stable stage.  Closure under subformulae ensures that every premise still uses
only formulae in $\Phi$.  Alongside each member of the basis at stage $n$, we
retain a proof of height at most $n$.  At stage $0$ this is the corresponding
nullary proof.

Suppose proofs have been retained for the basis of $U_n$.  Whenever the
computation of predecessors retains a conclusion
\[
 C\Longrightarrow_r(P_1,\ldots,P_k)
\]
above a tuple of basis elements $B_i\preceq P_i$, the finite search in
\cref{prop:bounded-sequents} supplies witnessing embeddings.  The effective weakening lemma turns the retained proof of each $B_i$ into a
proof of
$P_i$ of no greater height; appending the displayed inference gives a proof of
$C$ of height at most $n+1$.  Previously retained basis elements keep their proofs, and during
minimisation one proof is kept with each chosen representative.  Thus the
basis at every stage is effectively accompanied by proofs of the appropriate
height.

Let $N$ be the first stable index, and let $D_1,\ldots,D_m$ be the resulting
basis, with retained proofs $\pi_1,\ldots,\pi_m$.  By
\cref{thm:stabilisation}, their upward closure is exactly the set of provable
nested sequents over $\Phi$, and each $\pi_j$ has height at most $N$.  If $C$ is
provable, the basis property gives some $D_j\preceq C$; a witnessing embedding
is effectively found by \cref{prop:bounded-sequents}, and
\cref{lem:weakening-embedding} then constructs a proof of $C$ without increasing
height.  Conversely, every $D_j$ is provable, and \cref{lem:weakening-embedding} preserves
provability.
\end{proof}

The uniform proof-height bound is obtained only when the chain stabilises; it is not an elementary a priori complexity estimate.
Obtaining such an estimate would require length-function bounds for the
Kruskal order and independent bounds on the successive bases.  The argument
decides proof existence directly, without passing through a finite-model
property.

\section{Conclusion}

The obvious first route is a finite-countermodel argument.  The direct
quotient construction breaks when transitive closure creates a modal demand
that was absent before representatives were identified.  We avoid that
difficulty by staying with proofs: the finite objects are nested sequents, not
quotient worlds.

Kruskal's theorem makes the unbounded modal trees well-quasi-ordered, but the
algorithm still requires an effective computation of predecessors.  Since an inference may
occur in a context of arbitrary size, a separate argument is needed to show
that minimal predecessors have bounded representatives.
The finite-support lemma shows that such a representative can be chosen
within an explicit bound.  One keeps the nodes
used by the rule, the output, and the images of the basis elements, together with
the branch points forced by their rooted hull.  Once the required
input formulae have been retained at those nodes, unmarked unary stretches
between them may be suppressed.  The bound is not sharp; its role is only to make the predecessor search
finite.

Transitivity enters once more in the finite-support step.  The propagation
rules ask only for strict descendancy, which survives both the expansion of an
edge into a path and the suppression of unmarked degree-$2$ vertices on that
path.  A different path language may not survive the same reduction, so this
point would have to be checked afresh in any extension.  It holds for
$\IKfour$.  Over any finite set of formulae closed under subformulae, finitely
many cut-free proofs then suffice: weakening along an embedding supplies a
proof of every other provable nested sequent.

\end{document}